\documentclass[conference]{IEEEtran}
\IEEEoverridecommandlockouts

\usepackage[T1]{fontenc}
\usepackage{graphicx}
\usepackage{xcolor}
\usepackage{amsmath,amssymb,amsthm}
\usepackage{algorithm}
\usepackage{algorithmic}
\usepackage{booktabs}
\usepackage{multirow}
\usepackage{pifont}
\usepackage{soul}
\usepackage{balance}
\usepackage{stfloats}
\usepackage{xspace}
\usepackage[hidelinks]{hyperref}

\theoremstyle{plain}
\newtheorem{proposition}{Proposition}
\newtheorem{assumption}{Assumption}
\theoremstyle{remark}
\newtheorem*{remark}{Remark}

\providecommand{\cmark}{\ding{51}}
\providecommand{\xmark}{\ding{55}}

\newcommand{\myparatight}[1]{\smallskip\noindent{\bf {#1}:}~}
\newcommand{\sys}{\textsc{Centile}\xspace}

\makeatletter
\def\tablename{Table}
\long\def\@makecaption#1#2{%
\ifx\@captype\@IEEEtablestring%
\footnotesize\bgroup\par\centering\@IEEEtabletopskipstrut{\normalfont\footnotesize {#1.}\nobreakspace #2}\par\addvspace{0.5\baselineskip}\egroup%
\@IEEEtablecaptionsepspace
\else
\@IEEEfigurecaptionsepspace
\setbox\@tempboxa\hbox{\normalfont\footnotesize {#1.}\nobreakspace #2}%
\ifdim \wd\@tempboxa >\hsize%
\setbox\@tempboxa\hbox{\normalfont\footnotesize {#1.}\nobreakspace}%
\parbox[t]{\hsize}{\normalfont\footnotesize \noindent\unhbox\@tempboxa#2}%
\else%
\hbox to\hsize{\normalfont\footnotesize\hfil\box\@tempboxa\hfil}%
\fi\fi}
\makeatother

\begin{document}

\title{\sys: A Telemetry Foundation Model Evaluated by the Decisions It Drives}

\author{
	\IEEEauthorblockN{
		Zifan Zhang\textsuperscript{*},
		Zhichao Hou\textsuperscript{*},
		Tingxiang Ji,
		Yuchen Liu
	}
	\IEEEauthorblockA{
		Department of Computer Science, North Carolina State University, Raleigh, NC 27695, USA
	}
	\thanks{\textsuperscript{*}These authors contributed equally to this work.}
	
}

\maketitle

\begin{abstract}
Modern computing and networking infrastructure emits telemetry continuously, yet operators convert it into decisions with a separate predictor per task, entity, and horizon.
One generative model, pretrained once over an operator's own event streams, could replace this fleet, an approach that already scales to high-cardinality streams in recommendation systems.
However, point-forecast error on operational telemetry saturates near simple last-value baselines, so lower error alone need not improve the decisions it feeds.
To close this gap, we present \sys, a generative foundation model for network and systems telemetry, evaluated by replaying the decisions its calibrated conditional quantiles drive.
\sys treats heterogeneous telemetry as event-driven, irregularly timed entity streams and serves flexible forecast horizons in a single pass, requiring no future timestamps.
To our knowledge, \sys is the first pretrained telemetry model to improve both HPC scheduling and network provisioning decisions under replay, its runtime estimator transferring zero-shot across months and its pretrained weights across domains from hours of target data.
Extensive experiments on HPC job logs and network traffic confirm that \sys lowers the mean bounded slowdown of backfilling by up to approximately $77\%$ over deployed user estimates and roughly halves the deployed rule's violation rate.
Our code is available at \url{https://github.com/ZzZTripleZzZ/all-in-one}.
\end{abstract}

\section{Introduction}

Every consequential action in a production fleet begins as a forecast.
An HPC scheduler reserves thousands of nodes against a predicted job runtime, and a capacity planner locks in next-hour bandwidth at a predicted traffic percentile.
When the forecast is wrong, the cost is concrete, as idle reservations waste node-hours, under-provisioned links violate service agreements, and tail latency climbs at fleet scale.
The data for these forecasts is abundant, since every layer of modern computing and networking infrastructure emits telemetry continuously, from per-task lifecycle events and per-virtual-machine utilization in datacenters, to per-node power in supercomputers, to per-flow traffic volumes in networks spanning datacenter fabrics and national backbones~\cite{azure,m100,borg,borgng}.
These events arrive irregularly, on each entity's own clock rather than on a shared sampling grid.

The prevailing practice converts this telemetry into decisions one predictor at a time, with a separate model trained for each task, each entity type, and even each prediction horizon.
This is precisely the regime that foundation models displaced in language, vision, and recommendation, yet the telemetry prediction that underpins network performance analysis still lacks one.
We therefore ask whether systems telemetry can have its own foundation model, a single pretrained model that ingests any entity's history and serves any operational decision that consumes it.
Three requirements make this question concrete: the model must 1) serve any forecast horizon without retraining; 2) consume irregularly timed events as they arrive; and 3) carry what it has learned across domains, so that telemetry from one infrastructure can bootstrap decisions on another.
One pretrained model would retire this fleet, and a newly deployed system with only hours of history could still make calibrated decisions on day one.
Building the first such model is the goal of this paper.

\myparatight{Motivation and Challenges} Reaching this goal is difficult for reasons specific to networking systems.
First, the fleet itself is a recurring operational cost.
Each predictor carries its own training pipeline, retraining cadence, storage and inference footprint, and monitoring, an operational burden documented to dominate the lifecycle cost of production learning systems~\cite{paleyes}, and every new workload, entity type, or horizon adds another model and pipeline.
Second, the data itself is difficult to capture with a single model, as telemetry entities number in the millions and are heterogeneous, spanning users, jobs, machines, and flows.
Moreover, events do not arrive on a uniform clock, a substantial part of their predictive signal is carried by \textit{when} they happen, and operators query these streams at arbitrary input and output lengths, so a model that fixes its window and horizon must be duplicated per use.
Lastly, and most fundamentally, today's models are judged by a metric that is insensitive to what matters.
\textit{Actions speak louder than words}: a telemetry model exists to drive schedules and provisioning plans, yet it is ranked by point-forecast error, a number that saturates near last-value baselines on operational traces and can stay flat while the decisions it feeds diverge~\cite{dfl}.

\myparatight{Limitations and Opportunities} The foundation models proposed for networking so far each cover one domain, one modality, and usually one task.
Masked packet and flow encoders reconstruct held-out telemetry yet are never evaluated against a scheduler or a provisioner~\cite{etbert,netfound}, wireless models remain confined to their own air interface~\cite{mobigpt,lwm}, and language-model adaptations inherit billion-parameter inference costs that no per-decision millisecond budget can absorb~\cite{netllm}.
General-purpose time-series foundation models~\cite{chronos,timesfm} come closest, delivering one model for any horizon, but they pretrain on generic corpora rather than system telemetry, assume a uniform sampling grid that discards event timing, and serve fixed quantile grids whose most conservative level falls short of what provisioning requires.
None is at once \textit{generative}, \textit{pretrained on system telemetry itself}, and \textit{evaluated on the decisions operators actually run}.
Our key insight is that operational decisions consume \textit{quantiles} rather than point forecasts, as a backfilling scheduler reserves with a walltime estimate and a provisioner sets capacity at a percentile.
A generative model that learns the calibrated conditional distribution of telemetry is therefore exactly the artifact both decisions need, even where its point error appears unremarkable.

\myparatight{Objectives and Contributions} To this end, we present \sys\footnote{Our code is available at \url{https://github.com/ZzZTripleZzZ/all-in-one}.}, a generative foundation model for network and systems telemetry that is evaluated throughout at the level of the operational decisions it drives.
\sys treats heterogeneous telemetry as event-driven entity streams and embeds each event's value, resource covariates, and inter-event gap.
From one shared design, the same architecture, objective, and training procedure at every deployment, it pretrains one model per domain, built around three telemetry-specific choices: intensity-preserving attention over irregularly timed event streams, a heavy-tailed mixture output, and direct multi-horizon decoding.
\sys exposes calibrated conditional quantiles that feed each decision, serving every horizon in a single pass rather than by autoregressive rollout and needing no future timestamps at submit time, as sketched in Fig.~\ref{fig:arch}.

\begin{figure*}[!t]
\centering
\vspace{-.1in}
\includegraphics[width=0.95\textwidth]{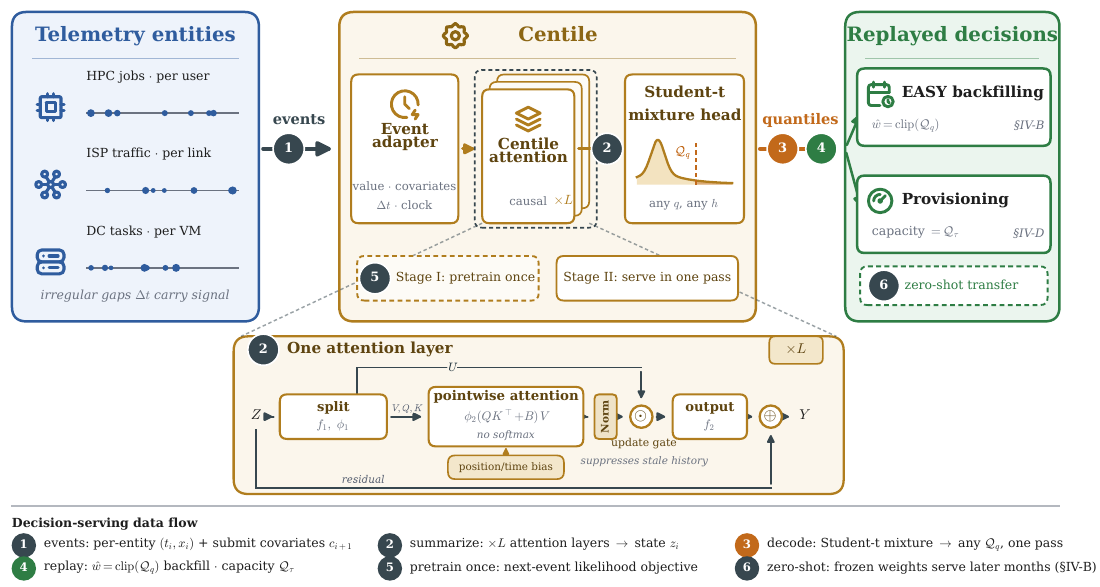}
\caption{\sys overview. Events from telemetry entities flow through the adapter, the pretrained attention layers, and the Student-t mixture head into calibrated quantiles that drive the replayed decisions. The inset expands one of the $L$ attention layers.}
\vspace{-.1in}
\label{fig:arch}
\end{figure*}
We instantiate this on two operational tasks that no prior telemetry model has jointly addressed: estimating job walltimes, the runtimes that Extensible Argonne Scheduling sYstem (EASY) backfilling schedules against, on production HPC job logs, and provisioning next-hour capacity on national ISP traffic.
A third domain, cloud virtual-machine telemetry, joins the pretraining traces for cross-domain transfer.
To our knowledge, this is the first work to show that one model design, pretrained on unlabeled telemetry, improves both scheduling and provisioning decisions under replay.
The estimator transfers zero-shot across months, and the pretrained weights carry over from traffic prediction to HPC scheduling.
Our main contributions are summarized as follows.

\begin{list}{\labelitemi}{\leftmargin=1em \itemindent=-0.08em \itemsep=.05em}
\item We recast network and systems telemetry prediction as one generative model over event-driven entity streams, and introduce a decision-replay protocol that scores a telemetry model by the scheduling and provisioning decisions its forecasts drive rather than by point-forecast error.
This protocol exposes differences that point-forecast metrics conceal.
\item Driving EASY backfilling from \sys quantiles lowers mean bounded slowdown over deployed user estimates by up to approximately $77\%$, and its estimator is the best walltime source on every evaluated month, including the unseen months served by the April-pretrained model.
\item Without any change to the design, \sys serves next-hour capacity provisioning on ISP traffic, roughly halving the violation rate of the deployed rule at comparable overprovisioning.
Its $1.3$M-parameter runtime estimator serves scheduling decisions in milliseconds on commodity CPUs.
\item Lastly, we conduct comprehensive decision-replay evaluations across a multi-month HPC job trace and a national ISP trace, spanning cross-month transfer, cross-domain transfer, pretraining scale, ablations, and deployment cost.
\end{list}

\section{Background and Positioning}
\label{sec:prelim}

\myparatight{Telemetry as entity event streams}
We model all telemetry through one abstraction.
A source exposes a set of \emph{entities} $\mathcal{E}$, where an entity may be a user, a compute node, a virtual machine, or a monitored link.
Each entity $e\in\mathcal{E}$ emits an ordered stream of timestamped \emph{events}
\begin{equation}
S_e=\big((t_1,x_1),(t_2,x_2),\dots\big),\qquad t_1<t_2<\cdots,
\label{eq:stream}
\end{equation}
where $x_i$ records the measurement observed at time $t_i$, such as a demand volume, a utilization sample, or an event duration, together with any covariates the source attaches.
The timestamps are the entity's own, so the inter-event gaps $\Delta_i=t_i-t_{i-1}$ vary over orders of magnitude rather than following a fixed sampling period.
Over such streams, operators pose two recurring questions: \emph{forecasting} the distribution of the next measurements at one or many horizons and, downstream of that forecast, the operational \emph{decisions} that consume its quantiles, from scheduling batch jobs to provisioning capacity.

\myparatight{Positioning}
Table~\ref{tab:position} places \sys against representative foundation models for systems and time series.
Its columns ask whether a model is generative in the sense of exposing a full predictive distribution, is pretrained on system telemetry, serves arbitrary output horizons, reuses one model design across tasks, and models irregular inter-event gaps.
No prior model is at once generative, pretrained on system telemetry, horizon-agnostic, reusable across tasks, and time-aware, and \sys is the first to satisfy all five.

\begin{table}[t]
\centering
\caption{Positioning against foundation models for systems and time series. \cmark: yes, \xmark: no, $\sim$: partial.}
\label{tab:position}
\setlength{\tabcolsep}{4pt}
\resizebox{\linewidth}{!}{%
\begin{tabular}{lcccccc}
\toprule
Model & Domain & Gen. & Sys-pretr. & Any-hor. & Multi-task & Time-aware \\
\midrule
ET-BERT~\cite{etbert}, netFound~\cite{netfound}   & packets    & \xmark & \cmark & \xmark & \xmark & \xmark \\
NetLLM~\cite{netllm}                               & networking & \cmark & \xmark & \xmark & \cmark & \xmark \\
MobiGPT~\cite{mobigpt}, UoMo~\cite{uomo}           & wireless   & $\sim$ & \cmark & $\sim$ & \cmark & \xmark \\
LWM~\cite{lwm}                                     & wireless   & \xmark & \cmark & \xmark & \xmark & \xmark \\
Chronos~\cite{chronos}, TimesFM~\cite{timesfm}     & generic TS & \cmark & \xmark & \cmark & \xmark & \xmark \\
\textbf{\sys (ours)}                               & HPC+ISP    & \cmark & \cmark & \cmark & \cmark & \cmark \\
\bottomrule
\vspace{-.2in}
\end{tabular}}
\end{table}

\myparatight{Key observation: telemetry is bursty in time yet correlated in value}
We measure, per entity, the coefficient of variation (CV) of inter-event gaps and the lag-one autocorrelation of the target across five public traces (Table~\ref{tab:obs}).
Inter-event timing varies widely, as job and task logs~\cite{fdata,borg,alibaba} reach a CV between one and four and node sampling punctuated by downtime~\cite{m100} reaches nearly twelve, while only the fixed-rate Azure sampler~\cite{azure} attains the zero of a strict grid.
At the same time, targets are frequently autocorrelated, up to $0.93$.
Both signals, the value history and the irregular timing, are available to a model that reads the stream in order yet discarded by a uniform-grid forecaster, which motivates a time-aware generative model of telemetry.

\begin{table}[b]
\centering
\vspace{-.1in}
\caption{Telemetry is irregular yet predictable. Median per-entity inter-event CV and lag-one autocorrelation. Larger CV means less clock-like timing.}
\label{tab:obs}
\setlength{\tabcolsep}{5pt}
\begin{tabular}{llcc}
\toprule
Dataset & Type & Inter-event CV & Autocorr. \\
\midrule
F-DATA (HPC jobs)~\cite{fdata}    & job log  & 3.97  & 0.48 \\
Borg (DC tasks)~\cite{borg}      & task log & 3.08  & 0.82 \\
Alibaba (DC tasks)~\cite{alibaba}   & task log & 1.03  & $-0.06$ \\
Azure (VM CPU)~\cite{azure}       & sampling & 0.00  & 0.09 \\
M100 (node power)~\cite{m100}    & sampling & 11.94 & 0.93 \\
\bottomrule

\end{tabular}
\end{table}

\myparatight{From forecasts to decisions}
We specify how each replayed decision consumes a served quantile.
Under EASY backfilling~\cite{lifka}, jobs queue in arrival order, and the scheduler makes a reservation for the queue head by projecting when running jobs will release their nodes from the walltime estimates those jobs carry~\cite{tsafrir}.
A waiting job may jump ahead of the queue head only if its own estimate indicates it will finish before the reservation, or if it fits entirely on nodes the reservation leaves spare.
The estimates therefore govern both the reservation and every backfill decision.
We score a schedule by \emph{bounded slowdown}~\cite{carastan},
\begin{equation}
\mathrm{BS} = \max\!\left(\frac{\text{wait} + \text{run}}{\max(\text{run},\,\epsilon)},\, 1\right),
\label{eq:bs}
\end{equation}
where $\text{wait}$ and $\text{run}$ are a job's queue wait and running time, and $\epsilon$ is a short floor that keeps very short jobs from dominating the ratio, with lower values preferred.
Capacity provisioning consumes a forecast differently, as the capacity for the next interval is set to the served quantile $\mathcal{Q}_\tau$ of predicted demand, and the operator moves $\tau$ to trade overprovisioning against the rate at which demand exceeds the provisioned capacity.
In both cases, we evaluate by \emph{decision replay}, re-executing the recorded decision process on the historical trace with the served estimate substituted for the deployed one, holding arrivals, true runtimes, and job sizes fixed, and scoring the decisions that result.

\section{The \sys Framework}
\label{sec:method}

\sys is a generative foundation model over per-entity telemetry event streams, assembled from three composable stages: a pluggable input adapter that embeds each event, a shared set of attention layers that summarize an entity's history, and a distributional head that emits the next measurement as a full predictive distribution.
We design \sys to learn the conditional law of the next event an entity emits, $p(x_{i+1}\mid x_{\le i})$, so that it can roll a stream forward to whatever horizon an operator queries.
Prior networking models instead use masked traffic encoders~\cite{etbert,netfound} or frozen language models adapted to networking tasks~\cite{netllm}, neither of which exposes such a distribution.

\myparatight{Design goals} Four goals shape the framework:
1) \textit{one artifact per domain}, meaning any entity history and any horizon within a domain are served by a single set of weights, and a new domain reuses the design unchanged;
2) \textit{calibrated distributions}, meaning the served quantiles carry their nominal coverage so that a decision policy can choose its operating point a priori;
3) \textit{submit-time causality}, meaning every prediction is computable from information available when the decision fires, with no future values or timestamps;
and 4) \textit{negligible inference cost}, meaning milliseconds per decision on the commodity CPUs that already run the scheduler and the provisioning loop.
Together, these goals operationalize the introduction's three requirements while adding the calibration and inference-cost constraints that an operational control plane imposes.

\myparatight{A High-Level Walkthrough} Figure~\ref{fig:arch} sketches the pipeline, and its numbered steps trace one decision end to end.
The input adapter maps each incoming event to a $D$-dimensional embedding of its value and the covariates of its deployment mode (step 1).
The $L$ attention layers aggregate the causal history and return a state that summarizes the entity up to the current event (step 2).
The distributional head turns this state into a Student-t mixture, from which any quantile is estimated empirically from $N_s$ Monte Carlo draws, as the mixture admits no closed-form inverse (step 3).
The served quantile becomes the walltime estimate or the provisioned capacity of the replayed decision (step 4).
In forecaster mode, a horizon feature conditions the head, so one pass yields a predictive distribution for every horizon at once.
In estimator mode, the model instead reads the final position of the current job's history and serves one clipped quantile as the walltime estimate.
The weights are pretrained once (step 5) and reused on later months zero-shot (step 6).

\subsection{Problem Formulation}
\label{ssec:problem}

A telemetry source exposes a set of entities, and each entity emits an ordered stream $S_e$ of timestamped events as in \eqref{eq:stream}.
Writing $\mathcal{H}_i=(x_{\le i},\, t_{\le i+1},\, c_{i+1})$ for what the control plane holds when the next event is submitted, \sys models the one-step predictive distribution
\begin{equation}
p\!\left(x_{i+1}\mid \mathcal{H}_i\right),
\label{eq:target}
\end{equation}
where $x_{\le i}$ is the observed value history, $t_{\le i+1}$ collects the timestamps through the submit time of the next event, and $c_{i+1}$ holds the submit-time covariates the scheduler already knows before the event completes, such as the requested node count and the user-declared limit.
In forecaster mode, the same construction is conditioned additionally on a requested horizon $h$, which \S\ref{ssec:arch} realizes without rollout.
Writing $\hat{F}_i$ for the cumulative distribution function of \eqref{eq:target} and $\mathcal{Q}_q(\hat{F}_i)$ for its $q$-quantile, the model is \emph{calibrated at level} $q$ when the served quantile carries its nominal coverage,
\begin{equation}
\Pr\!\left[\,x_{i+1} \le \mathcal{Q}_q(\hat{F}_i)\;\middle|\;\mathcal{H}_i\right] \;=\; q,
\label{eq:calib}
\end{equation}
a property whose aggregate consequence we audit on held-out data by comparing nominal levels with empirical coverage~\cite{gneiting}.
Calibration is the property both decisions consume, rather than point accuracy.
For walltime estimation, the predictive distribution is reduced to one served value by a quantile-and-clip rule,
\begin{equation}
\hat{w} = \mathrm{clip}\!\left(\mathcal{Q}_q\big(\hat{F}_i\big),\, w_{\min},\, w_{\mathrm{user}}\right),
\label{eq:clip}
\end{equation}
where $w_{\min}$ is a short minimum-runtime floor, for instance $60$ seconds, and $w_{\mathrm{user}}$ is the user-declared limit, which also remains the \emph{kill limit} at which an overrunning job is terminated.
For capacity provisioning, the reduction is simpler, as the capacity provisioned for the next interval on a link is the served quantile $\mathcal{Q}_\tau(\hat{F}_i)$ of predicted demand.

\subsection{Model Architecture}
\label{ssec:arch}

\myparatight{Input adapter} Each event arrives as a feature vector that the adapter embeds linearly into width $D$.
In estimator mode, the vector carries the entity's previous observed runtime, lagged and log-compressed, together with the job's submit-time request, its node count and declared limit, and an extended configuration adds the logarithm of the inter-arrival gap and clock features, the sine and cosine of the hour of day and the normalized day of week.
In forecaster mode, the adapter pairs a projection of each observed value with a learned embedding of its quantile bin, one of $K_c$ (e.g., $256$) bins spaced by the training-region quantiles, which adds a discrete regime signal to the continuous magnitude.
Targets are standardized in log space, keeping the heavy-tailed job runtimes and link traffic volumes on a scale the head models well.
Because the adapter is the only source-specific stage, a new telemetry source, whether a switch counter feed or a per-virtual-machine utilization stream, is added by declaring its feature vector while the stages above it stay fixed.

\myparatight{Intensity-preserving attention} The $L$ attention layers are built around a telemetry-specific principle, namely that the absolute intensity of activity is itself informative.
A standard Transformer's softmax attention~\cite{vaswani} renormalizes every history to sum to one, which flattens a submission burst and an idle period into the same relative weights.
Each layer therefore scores the causal past pointwise~\cite{hstu} and lets the aggregate scale with the entity's level of activity.
Replacing the softmax mixer is an established route to new attention properties, such as robustness~\cite{protransformer}, and here it targets intensity preservation.
From the layer input $Z$, four projections produce a gating branch, a value branch, and the query--key pair,
\begin{equation}
U(Z),\,V(Z),\,Q(Z),\,K(Z) \;=\; \mathrm{Split}\!\left(\phi_1\!\left(f_1(Z)\right)\right),
\label{eq:hstu-split}
\end{equation}
where $f_1$ is a linear map and $\phi_1$ is the sigmoid linear unit (SiLU) nonlinearity.
The layer then weights the causal past through pointwise SiLU attention scores, forgoing softmax normalization, and gates the pooled values elementwise,
\begin{equation}
Y \;=\; f_2\!\left(\mathrm{Norm}\!\left(\phi_2\!\left(QK^{\top} + B\right)V\right)\odot U\right),
\label{eq:hstu-agg}
\end{equation}
where $\phi_2$ is again SiLU, $\odot$ denotes the elementwise product, $\mathrm{Norm}$ is layer normalization, and $f_2$ is the output projection.
The learned bias $B$ carries a relative position term that weights events by their order in the stream and an optional time term, multi-scale Fourier features of the causal gap $\log(1+t_i-t_j)$ projected per head, which forecaster mode enables throughout so that the layer reads irregular spacing directly.
The four branches divide the computation, where $Q$ and $K$ select which past events are relevant to the one being processed, $V$ carries what each past event contributes to the summary, and $U$ decides how much of the pooled history the current event admits.
On a job stream, for instance, a submission burst raises the pairwise scores among its own events and the aggregate grows with the burst, while after a long idle period the gate admits little of the now-stale history.
A causal mask confines every position to its own past, and LayerScale~\cite{layerscale} stabilizes optimization at depth.
One unmodified layer stack thus serves both deployment modes, and the defaults in \S\ref{sec:exp} state which timing inputs each replay serves.
Inference cost scales as the attention product over a context of bounded length, so the per-decision compute is a small constant, as the measured milliseconds of \S\ref{sec:exp-dep} confirm, and the estimator runs on the CPU-only control plane that hosts the scheduler.

\myparatight{Distributional head} The head parameterizes a mixture of $M$ Student-t components~\cite{mdn}, an expressive form for the heavy right tails that job runtimes and traffic volumes exhibit.
From the state $z_i$, the head emits mixture weights $\alpha_m$, locations $\mu_m$, scales $\sigma_m$, and degrees of freedom $\nu_m$ for $m{=}1,\dots,M$, and assembles the predictive density
\begin{equation}
p(x_{i+1}\mid z_i) = \sum_{m=1}^{M}\alpha_m\,\mathcal{T}\!\left(x_{i+1};\mu_m,\sigma_m,\nu_m\right),
\label{eq:mixture}
\end{equation}
where a softmax makes the weights $\alpha_m$ sum to one, a softplus keeps each scale $\sigma_m>0$, and a positive offset keeps each $\nu_m>2$ so the component variance stays finite.
Each component is the location-scale Student-t density
\begin{equation}
\mathcal{T}(x;\mu,\sigma,\nu) = \frac{\Gamma\!\big(\tfrac{\nu+1}{2}\big)}{\Gamma\!\big(\tfrac{\nu}{2}\big)\sqrt{\pi\nu}\,\sigma}\left[1+\frac{1}{\nu}\!\left(\frac{x-\mu}{\sigma}\right)^{2}\right]^{-\frac{\nu+1}{2}},
\label{eq:studentt}
\end{equation}
where $\Gamma$ is the gamma function and a smaller $\nu$ yields a heavier tail than a Gaussian of the same scale, letting the head place mass on long runtimes and rare traffic spikes.
Unlike a quantile-regression head~\cite{mqrnn}, which fixes a grid of levels at training time and must be refit when the decision moves its operating point, the mixture head represents the full distribution once and serves any quantile on demand.

\myparatight{Direct multi-horizon conditioning} In forecaster mode, the horizon $h$ enters as Fourier features of $\log(1+h)$ that are projected and added to the state, so a single pass produces the distribution at every horizon.
Writing $z_i$ for the state at event $i$, the head reads a horizon-shifted state
\begin{equation}
\tilde{z}_i(h) = z_i + W_\psi\,\psi\!\left(\log(1+h)\right),
\label{eq:horizon}
\end{equation}
where $\psi$ maps the scalar $\log(1+h)$ to sinusoidal Fourier features and $W_\psi$ projects them to width $D$ before the head decodes $\tilde{z}_i(h)$.
Because the horizon enters only through this additive term, the layers run once and each requested horizon costs one more projection rather than one more rollout step, and \S\ref{ssec:theory} quantifies the gap to autoregressive rollout.

\subsection{Training and Inference}
\label{ssec:protocol}

\sys is pretrained once in an offline pipeline and then queried once per decision, which keeps the data-intensive stage off the control loop the scheduler and the provisioner run.
Pretraining groups the archived telemetry into per-entity streams in submission order and scores every event against the history that preceded it, so each position predicts its entity's next measurement from its own causal past by minimizing
\begin{equation}
\mathcal{L}(\theta) = -\frac{1}{|\mathcal{B}|}\sum_{i\in\mathcal{B}} \log p_\theta\!\left(x_{i+1}\mid z_i\right),
\label{eq:nll}
\end{equation}
where $\mathcal{B}$ indexes the scored positions in a batch, $\theta$ collects the adapter, attention, and head parameters, and $p_\theta$ is the mixture of \eqref{eq:mixture} evaluated at the realized measurement $x_{i+1}$.
The log-likelihood is a strictly proper scoring rule~\cite{gneiting}, so its minimizer over a sufficiently rich family is the true conditional distribution rather than any single point summary, which is what a decision consuming a quantile needs.
After training, the weights and the log-space normalization statistics are frozen together, so a later deployment window enters the same scale the head was fit on.
At inference, the estimator fetches the entity's most recent $C$ events, appends the submit-time covariates $c_{i+1}$, and runs one pass through the model.
It then estimates $\mathcal{Q}_q(\hat{F}_i)$ by Monte Carlo sampling from the mixture and returns the clipped estimate of Eq.~\eqref{eq:clip} within a single scheduling cycle.
Training reads only data strictly earlier than the evaluation window, and each entity's history grows in submit order at inference, so no prediction ever reads its own future, a contract every baseline in \S\ref{sec:exp} obeys as well.

\myparatight{Pretraining and zero-shot transfer} Because the estimator is pretrained once rather than fit per deployment, a single artifact can be reused across time.
We pretrain on earlier months and deploy the frozen artifact on later months, which supply only their own per-entity histories as context, so successive months are served as they arrive without retraining.

\subsection{Why Calibrated Quantiles Improve Decisions}
\label{ssec:theory}

In this subsection, we provide a theoretical account of the design by establishing how the calibration property of Eq.~\eqref{eq:calib} converts into guarantees on the two replayed decisions and why direct decoding avoids the compounding penalty of rollout at long horizons.
The analysis rests on two assumptions that are standard in probabilistic forecasting and in the analysis of learned dynamical models~\cite{gneiting,dad,hewamalage}, and both are attainable in deployment rather than idealizations.
Calibration is what the objective of Eq.~\eqref{eq:nll} already pursues, since a strictly proper scoring rule is minimized by the conditional law itself, and the served level is audited against measured coverage before deployment rather than assumed.
The rollout comparison needs only a bounded expected one-step error and a Lipschitz feedback map.

\begin{assumption}[Quantile calibration]
\label{as:calib}
At the served level and horizon, the served distribution satisfies Eq.~\eqref{eq:calib} given the prediction-time history, for the events the decision consumes.
\end{assumption}

Under Assumption~\ref{as:calib} at level $\tau$, demand exceeds the capacity provisioned at the served quantile at the promised rate $1-\tau$.
The classical newsvendor argument then identifies the cost-optimal level as the fractile $\kappa_u/(\kappa_u+\kappa_o)$ of the per-unit violation cost $\kappa_u$ and overprovisioning cost $\kappa_o$, since the expected cost is convex in the provisioned capacity and its derivative changes sign at that fractile~\cite{banrudin}.
Calibration at the served level is what carries this first-order condition from the unknown demand law over to the distribution \sys serves.
Thus, the served quantile becomes a contract, and \S\ref{sec:exp-ces} tests it empirically, where measured violation rates stay at or below the promised $1-\tau$ at every served level while a miscalibrated baseline exceeds it at two of three served levels.

\begin{proposition}[Reservation reliability under EASY]
\label{prop:easy}
Consider the $J$ jobs backing a reservation, each with a walltime estimate served by Eq.~\eqref{eq:clip} at level $q$ upon submission.
Under Assumption~\ref{as:calib}, the probability, assessed at submission, that any of them overruns its served estimate is at most $J(1-q)$, and the clipping in Eq.~\eqref{eq:clip} never degrades this coverage.
\end{proposition}

\begin{proof}[Proof sketch]
For one job, if $w_{\min} \le \mathcal{Q}_q \le w_{\mathrm{user}}$ the served estimate equals $\mathcal{Q}_q$ and the overrun probability is $1-q$ by Eq.~\eqref{eq:calib}.
Otherwise the estimate is $w_{\mathrm{user}}$, which the runtime cannot exceed because the kill limit enforces it, so the overrun probability is zero.
Raising the estimate to the floor $w_{\min}$ only enlarges coverage.
A union bound over the $J$ backing jobs gives the claim.
\end{proof}

In other words, with a calibrated $q$, the projected node-release times behind a reservation are simultaneously conservative with probability at least $1-J(1-q)$, a guarantee that becomes operative at conservative served levels where $J(1-q)$ stays small.
Read as a sizing rule, an operator targeting reservation reliability $1-\delta$ serves at $q \ge 1-\delta/J$.
At the same time, the estimates stay well below the deployed limits and keep backfill holes open.
The architecture of \S\ref{ssec:arch} enters here, since the bound is only as useful as the distribution behind it.
The mixture head keeps a heavy right tail per job rather than one global margin, and the intensity-preserving attention supplies the history that separates a user's long jobs from the short ones.
Thus, the conditional quantile is conservative precisely for the jobs whose histories predict long runtimes and aggressive elsewhere, a per-job adaptivity no single point estimate provides, and the replay of \S\ref{sec:exp-bf} measures its value at moderate served levels.
Capacity provisioning consumes the same guarantee in a different currency, where $1-\tau$ is the fraction of intervals in which a link carries more than the capacity provisioned for it, so the served level states a service-level target rather than a tuned constant.
Direct decoding carries a guarantee of its own, under the standard regularity for autoregressive prediction, namely a bounded expected one-step error and a Lipschitz feedback map~\cite{dad,hewamalage}.
Rollout feeds each prediction back as input, so the error at every step decomposes into a fresh one-step term plus the amplified error it inherited.
Unrolling this recursion bounds the expected rollout error at horizon $h$ by a geometric sum in the Lipschitz constant of the feedback map, a bound that stays finite only under a strict contraction, grows linearly in $h$ when the constant is one, and grows exponentially beyond it.
Bursty telemetry offers no certificate of contraction~\cite{dad}.
The direct head of Eq.~\eqref{eq:horizon} removes the compounding term entirely, as each horizon is a separate target read from the same observed history, so no prediction is ever fed back and each horizon incurs only its own estimation error.
Thus, the widening gap measured in \S\ref{sec:exp-abl} is consistent with the compounding the bound predicts rather than a tuning effect.

\begin{remark}
The guarantees are conditional on measured calibration rather than unconditional properties of the architecture, and the rollout comparison rests on the stated regularity of the feedback map.
Once measured, calibration makes the guarantees above operative up to the Monte Carlo error of quantile estimation.
\end{remark}

\section{Experiments}
\label{sec:exp}

\subsection{Experimental Setup}

We evaluate \sys at the decision level on two tasks: 1) walltime estimation for EASY backfilling on F-DATA~\cite{fdata}, production job logs from the Fugaku supercomputer; and 2) capacity provisioning on CESNET-TimeSeries24~\cite{cesnet}, hourly traffic volumes from a national ISP backbone network.
The primary F-DATA month contains roughly $117$K jobs, of which the replay window covers $48{,}826$ jobs from $298$ users over two weeks, following a one-week warm-up that fits the estimator and seeds per-user histories.
The median user-declared limit, six hours, exceeds the median true runtime, roughly $21$ minutes, by more than an order of magnitude.
CESNET is split chronologically at $80/20$ per series and evaluated on $3{,}000$ held-out windows.

\myparatight{Decision replay}
Each F-DATA month is replayed through an EASY backfilling scheduler in which the walltime estimate is the \textit{only} quantity that varies between the compared estimators, and the kill limit remains the user-declared value.
Reservations apply the standard estimate-extension correction~\cite{tsafrir}, extending an expired estimate in place rather than killing the job.
Cluster capacity is set from the offered load of the replayed window, and we evaluate the backlogged band (offered-to-capacity ratios between $0.95$ and $1.46$), where queues form.
We report the mean and the 95th percentile of bounded slowdown as defined in \S\ref{sec:prelim}.
Forecast quality is reported through per-series Mean Absolute Scaled Error (MASE)~\cite{m4}, for which $1.0$ matches the naive last-value forecast, pinball loss at the $0.99$ quantile~\cite{mqrnn}, and a nine-quantile Continuous Ranked Probability Score (CRPS) approximation~\cite{gneiting}.

\myparatight{Estimator baselines}
We compare five walltime sources: 1) the user-declared limit, the operator status quo; 2) the per-user rolling 90th-percentile of past runtimes (History q90); 3) the average of the user's last two runtimes (Tsafrir last-2), the canonical system-generated predictor~\cite{tsafrir}; 4) gradient-boosted quantile regression on causal tabular features (GBM quantile), the strongest non-sequential learned baseline in recent scheduling studies~\cite{gaussier,uarp}; and 5) \sys, serving the requested quantile of its predicted runtime distribution at submit time.
Scheduling with true runtimes is a reference, and all walltime sources observe identical causal histories.

\myparatight{Implementation and defaults}
Models are implemented in PyTorch and trained with \texttt{bf16} autocast on one A30 GPU, and per-user histories are maintained strictly in submission order during both training and replay.
The runtime estimator uses width $D{=}256$, $L{=}4$ attention layers, eight attention heads, a context of $C{=}96$ events, $M{=}6$ mixture components, $N_s{=}2048$ Monte Carlo draws per served quantile, and $30$ training epochs, serving $q{=}0.5$ in the primary tables and $q{=}0.6$ in the capacity and month sweeps, and the CESNET forecaster uses $D{=}512$ and $L{=}8$.
\sys rows in Table~\ref{tab:backfill} report the mean and standard deviation over five seeds.
The primary and month replays serve the estimator in its minimal configuration, which reads the previous runtime and the request covariates without event-timing inputs, as richer timing inputs did not improve this replay.
The transfer and pretraining-scale protocols serve the extended covariates and the attention time term, and the arms of \S\ref{sec:exp-abl} are covariate-matched as noted there.

\subsection{Backfilling Replay on F-DATA}
\label{sec:exp-bf}

\begin{table}[t]
\centering
\caption{Bounded slowdown by walltime source at two backlogged capacities on the primary month. Bold marks the best estimator, and true-runtime scheduling is included only as a reference.}
\label{tab:backfill}
\setlength{\tabcolsep}{4pt}
\begin{tabular}{lcccc}
\toprule
& \multicolumn{2}{c}{$20$K nodes} & \multicolumn{2}{c}{$23$K nodes} \\
\cmidrule(lr){2-3}\cmidrule(lr){4-5}
Walltime source & mean & P95 & mean & P95 \\
\midrule
User estimates       & 94.8 & 279.9 & 53.1 & 171.4 \\
History q90          & 56.6 & 104.8 & 37.6 & 53.9 \\
Tsafrir last-2       & 24.3 & 31.1  & 23.1 & 21.2 \\
GBM quantile         & 25.4 & 34.9  & 19.7 & 32.1 \\
\textbf{\sys}        & \textbf{22.2}$\pm$2.1 & \textbf{28.0} & \textbf{15.5}$\pm$1.1 & \textbf{17.8} \\
\midrule
\textit{True runtimes (ref.)} & \textit{17.7} & \textit{69.6} & \textit{28.9} & \textit{169.9} \\
\bottomrule
\vspace{-.1in}
\end{tabular}
\end{table}

\myparatight{\sys schedules near, and beyond, the true-runtime reference}
Table~\ref{tab:backfill} reports the primary month at two capacities in the backlogged band.
It is observed that \sys attains the lowest mean bounded slowdown among all estimators at both capacities, $22.2$ and $15.5$ versus $24.3$ and $23.1$ for the strongest history-based estimator, respectively, improving on the deployed user estimates by up to approximately $77\%$.
At the $23$K-node capacity \sys also surpasses scheduling with \textit{true runtimes} ($15.5$ versus $28.9$), and its P95 tail is the best of every walltime source, reference included, at both capacities.
This is because schedule quality under EASY is not monotone in estimate accuracy, a property established well before learned estimators~\cite{tsafrir}.
Estimates far above the true runtime, as the deployed user limits are, block backfilling outright because an inflated request no longer fits the holes a reservation leaves.
Exact runtimes sit at the opposite extreme, where reservations receive no protective slack and the scheduler forfeits the heel-and-toe margin by which slightly conservative estimates keep short jobs flowing into shrinking holes~\cite{tsafrir}.
The optimum therefore sits between the two, at estimates that are accurate yet slightly conservative for the jobs that back reservations, which is precisely the operating point a calibrated conditional quantile provides.
Consequently, the true-runtime reference does not bound the schedule quality attainable under EASY, and it is not monotone in cluster size either, since the backfill opportunities that carry the effect grow with the free capacity.
Its tail makes the same point, as the reference P95 rises from $69.6$ to $169.9$ across the two capacities while every estimator improves, the signature of a few queue-head jobs repeatedly displaced rather than of a uniformly worse schedule.
Figure~\ref{fig:band} extends the replay to five capacities across the backlogged band at $q{=}0.6$, where \sys improves on the deployed estimates and the history baseline at every capacity and surpasses the true-runtime reference at the $23$K edge ($16.1$ versus $28.9$).

\begin{figure}[t]
\centering
\includegraphics[width=0.85\linewidth]{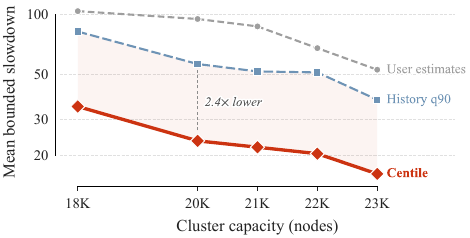}
\caption{Mean bounded slowdown across the backlogged capacity band on the primary month. \sys improves on the deployed estimates and the history baseline at every capacity.}
\label{fig:band}
\end{figure}

\begin{table}[b]
\centering
\caption{Mean bounded slowdown across F-DATA months at representative backlogged operating points. $\dagger$ marks months served zero-shot by the April-pretrained estimator.}
\label{tab:months}
\setlength{\tabcolsep}{3.5pt}
\begin{tabular}{lccccc}
\toprule
Walltime source & Jun & Jul$^\dagger$ & Aug & Sep & Oct$^\dagger$ \\
\midrule
User estimates & 3.68 & 66.3 & 31.7 & 21.9 & 14.0 \\
History q90    & 2.54 & 58.2 & 10.1 & 14.2 & 6.09 \\
Tsafrir last-2 & 3.32 & 58.1 & 9.08 & 10.5 & 5.85 \\
GBM quantile   & 3.35 & 54.4 & 14.6 & 11.7 & 4.15 \\
\textbf{\sys}  & \textbf{1.86} & \textbf{42.6} & \textbf{7.35} & \textbf{7.59} & \textbf{3.78} \\
\bottomrule
\end{tabular}
\end{table}

\begin{figure}[!b]
\centering
\includegraphics[width=0.85\linewidth]{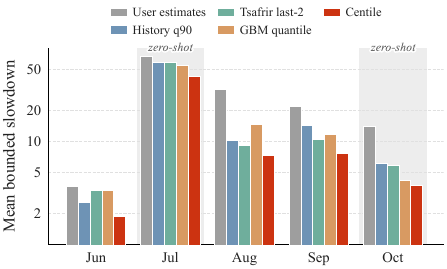}
\caption{\sys attains the lowest mean bounded slowdown on every month. Shaded months are served zero-shot by the April-pretrained model.}
\label{fig:months}
\end{figure}

\myparatight{One pretrained estimator transfers across months}
Table~\ref{tab:months} and Fig.~\ref{fig:months} extend the replay to five further months at representative backlogged operating points.
We observe that \sys is the best estimator in every column.
The July and October columns are served by the \textit{zero-shot} model, pretrained on April and never shown the target month, while the remaining columns fit the estimator on their own warm-up windows.
This advantage can be attributed to the pretraining data, which is larger and more diverse than any single month's warm-up window.
The advantage holds at the P95 tail in every column, adding no tail risk.

\myparatight{Impact of pretraining scale}
Enlarging the pretraining data from roughly $148$K to $1.07$M jobs lowers the mean bounded slowdown at the moderate-overload point by approximately $35\%$ ($11.9$ to $7.7$).
The largest pretraining set of roughly $1.8$M jobs attains the best light-load slowdown ($2.4$, roughly $25\%$ below the smallest set).

\subsection{Cross-Domain Transfer}
\label{sec:exp-transfer}

\myparatight{Setup}
The zero-shot months above hold the domain fixed, so we now ask whether the pretrained weights transfer across domains.
We pretrain the estimator design in its extended configuration on source traces of per-institution hourly traffic and per-virtual-machine CPU telemetry, neither of which contains a single job duration, cast into the estimator's event format with per-domain standardization.
These weights then initialize the F-DATA walltime estimator, fine-tuned on the last six hours, one day, two days, or all seven days of target history, then replayed at the $20{,}000$-node capacity as in \S\ref{sec:exp-bf}.
Three arms share every remaining setting: a scratch estimator trained from random initialization at each budget, a transfer estimator initialized from the cross-domain checkpoint, and a statistics-only arm that loads the checkpoint and fits only normalization statistics.
The history each arm may read at inference is identical, every entry is a three-seed mean with standard deviation, and the budget protocol yields numbers slightly above the five-seed results of Table~\ref{tab:backfill}.

\begin{table}[t]
\centering
\caption{Cross-domain transfer. Mean bounded slowdown at the $20{,}000$-node capacity (three-seed mean$\pm$std) as the target-data budget shrinks, with Transfer initialized from network and cloud telemetry alone. Bold marks the better arm.}
\label{tab:transfer}
\setlength{\tabcolsep}{6pt}
\begin{tabular}{lcc}
\toprule
Target budget & Scratch & Transfer \\
\midrule
$6$ hours       & $83.0\pm36.0$ & $\mathbf{35.9\pm1.3}$ \\
$1$ day         & $43.7\pm13.1$ & $\mathbf{22.8\pm2.3}$ \\
$2$ days        & $25.8\pm1.2$  & $\mathbf{24.6\pm2.9}$ \\
$7$ days (full) & $26.0\pm2.2$  & $26.0\pm1.5$ \\
\midrule
Statistics only, no training & \multicolumn{2}{c}{$38.2$} \\
\bottomrule
\vspace{-.1in}
\end{tabular}
\end{table}

\myparatight{Cross-domain pretraining makes scarce-data training reliable}
We observe that at the six-hour data budget the three scratch seeds are highly inconsistent, with mean bounded slowdown ranging from $36.7$ to $124.5$, so a single from-scratch training run at this budget is not auditable in advance.
The transfer arm is instead consistent across seeds, $35.9\pm1.3$ (Table~\ref{tab:transfer}), and lowers the served-quantile pinball loss by roughly $35\%$ at every budget below the full window.
This reliability can be attributed to the cross-domain checkpoint, which supplies an event-stream prior that random initialization cannot recover from a few hundred jobs, and a checkpoint pretrained on the network stream alone performs comparably.
As the budget grows the two arms converge, $24.6$ versus $25.8$ at two days and parity at seven days.

\myparatight{One day of target data suffices}
Notably, the transfer estimator fine-tuned on a single day of target history reaches $22.8$ versus $43.7$ for scratch, and the worst transfer seed ($26.1$) outperforms the best scratch seed ($33.9$).
One day of transfer already lands below the protocol's own seven-day arms and well below the deployed estimates ($94.8$) and the history baseline ($56.6$).
The statistics-only arm yields $38.2$, below both deployed references, because the transferred weights already encode the heavy-tailed event structure the decision consumes.

\subsection{Capacity Provisioning on CESNET}
\label{sec:exp-ces}

\begin{table}[t]
\centering
\caption{Provisioning violation rate (\%) at matched served quantiles on CESNET. \sys incurs fewer violations at every $\tau$ and serves quantile levels beyond the zero-shot model's grid.}
\label{tab:pareto}
\setlength{\tabcolsep}{5pt}
\begin{tabular}{lcccc}
\toprule
Served quantile $\tau$ & 0.5 & 0.8 & 0.9 & 0.95 \\
\midrule
Chronos-Bolt zero-shot~\cite{chronos} & 48.7 & 23.0 & 13.1 & --- \\
\textbf{\sys} & \textbf{42.8} & \textbf{17.5} & \textbf{9.0} & \textbf{4.7} \\
\bottomrule
\vspace{-.1in}
\end{tabular}
\end{table}

\begin{figure}[t]
\centering
\includegraphics[width=0.85\linewidth]{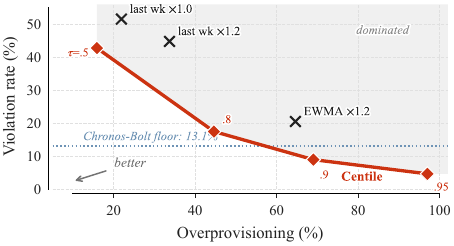}
\caption{Provisioning frontier on CESNET. Crosses denote operator rules, and the zero-shot model's grid ends at $\tau{=}0.9$.}
\vspace{-.1in}
\label{fig:pareto}
\end{figure}

\myparatight{Calibrated quantiles halve provisioning violations}
Figure~\ref{fig:pareto} replays next-hour capacity provisioning, sweeping the served quantile $\tau$ against the rules operators deploy in practice.
Notably, at an overprovisioning level comparable to the Exponentially Weighted Moving Average (EWMA) rule, \sys cuts the violation rate from $20.5\%$ to $9.0\%$, and every \sys operating point improves on the operator-rule frontier.
Table~\ref{tab:pareto} shows that \sys yields fewer violations than a zero-shot time-series foundation model at every $\tau$ the baseline serves, and only the calibrated model keeps its violation rate within the promised budget.
Furthermore, the native quantile grid of Chronos-Bolt~\cite{chronos} ends at $0.9$, so its violation rate cannot drop below $13.1\%$ at any served quantile, while the heavy-tail mixture head of \sys serves a calibrated $\tau{=}0.95$ and reaches $4.7\%$.
\begin{table}[t]
\centering
\caption{Forecast quality on CESNET against size-matched sequence models trained identically on the same streams, where GRU is a gated recurrent unit and TGR a gated-retention model (smaller values indicate better performance).}
\label{tab:dist}
\setlength{\tabcolsep}{2.6pt}
\resizebox{\linewidth}{!}{%
\begin{tabular}{lccccccccc}
\toprule
& \multicolumn{3}{c}{MASE} & \multicolumn{3}{c}{Pinball@0.99} & \multicolumn{3}{c}{CRPS} \\
\cmidrule(lr){2-4}\cmidrule(lr){5-7}\cmidrule(lr){8-10}
Model & $h{=}1$ & $h{=}4$ & $h{=}24$ & $h{=}1$ & $h{=}4$ & $h{=}24$ & $h{=}1$ & $h{=}4$ & $h{=}24$ \\
\midrule
GRU  & 0.958 & 1.079 & 1.086 & 0.031 & 0.037 & \textbf{0.035} & 0.561 & 0.631 & 0.623 \\
TGR  & 1.106 & 1.278 & 1.115 & 0.037 & 0.046 & 0.043 & 0.637 & 0.717 & 0.655 \\
\textbf{\sys} & \textbf{0.928} & \textbf{0.946} & \textbf{1.034} & \textbf{0.029} & \textbf{0.033} & \textbf{0.035} & \textbf{0.540} & \textbf{0.550} & \textbf{0.587} \\
\bottomrule
\end{tabular}}
\vspace{-.1in}
\end{table}

\myparatight{The served distribution is also the sharpest}
The provisioning gain could in principle come from the operating point alone, so Table~\ref{tab:dist} compares the served distribution itself against size-matched sequence baselines.
We observe that \sys attains the best MASE and CRPS at every horizon, $0.928$ versus $0.958$ MASE against the GRU at one hour, and the best or tied pinball loss at the $0.99$ quantile.
This can be attributed to the mixture head, which places explicit mass on the rare traffic spikes that a squared-error head averages away, and the advantage widens with the horizon.
These results underscore that the same design serves a second decision in a second domain without modification.

\subsection{Ablation: Where the Gains Come From}
\label{sec:exp-abl}

\myparatight{The choice of operating point is robust}
Sweeping the served quantile from $q{=}0.5$ to $q{=}0.8$ leaves \sys below both the deployed estimates and the history baseline at every setting on both primary capacities, and at $q{=}0.5$ and $q{=}0.6$ it outperforms every estimator.
Together with the measured calibration of Fig.~\ref{fig:calib}, this robustness lets the served quantile be chosen from the decision's risk preference alone.

\myparatight{Direct decoding outperforms rollout at an order of magnitude lower cost}
With identical weights on the size-matched gated-retention model, rollout error compounds to $8.34$ MASE by $64$ steps while direct decoding holds $4.35$, the gap widening monotonically with horizon in every seed.
One direct pass serves all $25$ input-length and horizon combinations in $5.2$ seconds, versus $69.4$ seconds for one rollout combination.

\myparatight{The sequence is the signal}
Both remaining ablations use a covariate-matched three-seed protocol, so their numbers fall slightly below the five-seed results of Table~\ref{tab:backfill}.
Removing the user's event history degrades the estimator past the operator status quo ($126.9$ versus $94.8$ at $20$K nodes), while removing the current job's covariates leaves it nearly unchanged ($20.2$ versus $20.0$), so the per-user sequence carries nearly all of the decision-relevant signal.

\myparatight{The attention layers are necessary}
Under this protocol, the \sys attention layers reach a mean bounded slowdown of $20.0\pm1.4$ and $15.5\pm0.3$ at the two primary capacities.
A size-matched GRU trained identically reaches $27.4\pm4.2$ and $20.7\pm2.7$, leaving \sys roughly a quarter lower with a three- to ninefold smaller seed standard deviation.
\begin{figure}[t]
\centering
\includegraphics[width=0.85\linewidth]{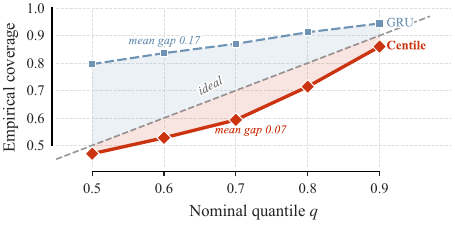}
\caption{Nominal quantile vs.\ empirical coverage of the runtime estimator against a protocol-matched GRU.}
\vspace{-.1in}
\label{fig:calib}
\end{figure}
Figure~\ref{fig:calib} explains the gap, as the nominal quantiles of \sys track their empirical coverage with a mean gap of $0.07$ versus the GRU's $0.17$.

\subsection{Deployment Considerations}
\label{sec:exp-dep}

\begin{table}[t]
\centering
\caption{Deployment profile of the \sys runtime estimator. All timing and size entries are measured.}
\label{tab:cost}
\setlength{\tabcolsep}{5pt}
\begin{tabular}{ll}
\toprule
Property & Measured value \\
\midrule
Parameters & $1.32$\,M \\
Artifact size (\texttt{bf16}) & $2.5$~MiB \\
Training, one A30 GPU & approx.\ $5$ minutes \\
Per-job CPU inference ($4$ cores) & $4.2$~ms batched, $7.7$~ms single \\
New-month deployment & zero-shot, no retraining \\
\bottomrule
\vspace{-.15in}
\end{tabular}
\end{table}

\myparatight{One artifact to operate}
Table~\ref{tab:cost} profiles the estimator, which adds negligible cost to a scheduling cycle operating at second granularity, with no GPU on the inference path.
A specialist approach maintains one model per task, per horizon, and per period, while \sys collapses this product into one $2.5$~MiB artifact per domain, three orders of magnitude below adapted billion-parameter models~\cite{netllm}.
Cold-user estimates degrade gracefully, clipping to the declared limit and affecting only $0.8\%$ of submissions with fewer than three prior observations.
Nothing in the design is specific to these two decisions, as any quantile-consuming policy inherits the same calibration and inference path.
In essence, a telemetry model's value is revealed where it is deployed, and \sys is designed to be evaluated there.

\section{Related Work}
\label{sec:related}

\myparatight{Foundation models for networking and systems} Pretrained traffic models such as ET-BERT~\cite{etbert}, netFound~\cite{netfound}, NetGPT~\cite{netgpt}, and NetMamba~\cite{netmamba} learn packet and flow representations, mainly for classification, while NetLLM~\cite{netllm} adapts a frozen language model to networking tasks and MobiGPT~\cite{mobigpt}, UoMo~\cite{uomo}, and LWM~\cite{lwm} target wireless networks.
Each remains within a single domain.

\myparatight{Time-series forecasting and foundation models} Deep forecasters are trained per dataset and per horizon, whereas Chronos~\cite{chronos}, TimesFM~\cite{timesfm}, Moirai~\cite{moirai}, Lag-Llama~\cite{lagllama}, and MOMENT~\cite{moment} pretrain once on broad corpora, and another line reprograms frozen language models into forecasters~\cite{timellm,onefitsall}.
Closest to our setting, the Tiny Time Mixer was evaluated zero-shot on CESNET traffic~\cite{ttm,ispforecast}, CLOUDADV aligns zero-shot forecasts with cloud instance sizing~\cite{cloudadv}, and decision-focused fine-tuning lowers feeder-dispatch cost~\cite{dffeeder}.
Yet, each treats inputs as uniformly-sampled series and evaluates one decision in one domain.
Because estimate accuracy governs backfilling performance, predictors draw on system-generated recent-history predictions~\cite{tsafrir}, machine learning~\cite{gaussier,klusacek}, uncertainty-aware runtime prediction~\cite{uarp}, and learned backfilling policies~\cite{rlbackfill}, all retrained per system and exposing a point estimate.

\myparatight{Learning over systems telemetry} Machine learning also drives operations tasks, from workload forecasting~\cite{azure,m100} and autoscaling~\cite{autopilot} to telemetry imputation~\cite{zoom2net}, traffic engineering~\cite{dote}, and system log representation~\cite{hlogformer}, but one task-specific model at a time.
Network digital twins evaluate control policies on a mirrored network before deployment~\cite{opentwin}, drive data-driven optimization of reliable edge caching~\cite{dtcache}, and revise the twin's ingested state as telemetry changes~\cite{untwinicdcs}, with their creation and optimization surveyed in~\cite{dnt}.
In situ video-streaming experiments reached the same conclusion, as proxy-metric gains failed to carry to the deployed decision~\cite{puffer}.
\sys instead pretrains one generative model on systems telemetry, replacing the per-task fleet with one design evaluated at the decisions it serves.

\section{Conclusion}

In this paper, we introduced \sys, a generative foundation model for network and systems telemetry evaluated at the level of the decisions it drives.
\sys treats telemetry as irregularly timed entity streams and serves every horizon as calibrated quantiles in a single pass from one pretrained model per domain.
Extensive experiments on Fugaku job logs and national ISP traffic confirm that this design lowers the mean bounded slowdown of backfilling by up to approximately $77\%$ over deployed user estimates and roughly halves the deployed rule's violation rate.
The $2.5$~MiB runtime estimator transfers zero-shot across months and serves in milliseconds on commodity CPUs.

\section{Use of AI Disclosure} AI tools~\cite{anthropic_claude} supported code development and manuscript editing, and the authors reviewed all generated content for accuracy.


\bibliographystyle{IEEEtran}
\bibliography{ref}

\end{document}